\documentclass[journal]{IEEEtran}

\usepackage{cite}
\usepackage{amsmath, amssymb, amsfonts, amsthm}
\usepackage{mathtools}
\usepackage{algorithmic}
\usepackage{graphicx}
\usepackage{textcomp}
\usepackage{booktabs}
\usepackage{soul}
\usepackage{empheq}
\usepackage{enumerate}
\usepackage{url}
\usepackage{mdframed}
\usepackage[firstpage]{draftwatermark}
\usepackage{xcolor}
\usepackage{hyperref}

\hypersetup{
	colorlinks,
	linkcolor={red!60!black},
	citecolor={blue!60!black},
	urlcolor={blue!80!black}
}

\DeclareMathOperator{\range}{\mathcal{R}}
\DeclareMathOperator{\trace}{\text{tr}}
\newcommand{\covpred}{\Sigma_{\scriptscriptstyle{\Delta}}}
\newcommand{\covphi}{\Sigma_{\scriptscriptstyle{\varphi}}}
\newcommand{\what}[1]{\widehat{#1}}
\newcommand{\real}{\mathbb{R}}
\newcommand{\yh}{\hat{\textrm{\bfseries y}}}
\newcommand{\That}{\what{\Theta}}
\newcommand{\zinit}{\textrm{\bfseries z}}
\newcommand{\yf}{{\textrm{\bfseries y}}}
\newcommand{\uf}{\textrm{\bfseries u}}

\newcommand{\wtilde}[1]{\widetilde{#1}}
\newcommand{\wbar}[1]{\overline{#1}}

\newmdtheoremenv{theorem}{Theorem}

\newtheorem{assumption}{Assumption}
\newtheorem{remark}{Remark}
\newtheorem{corollary}{Corollary}
\newtheorem{lemma}{Lemma}

\makeatletter
\newcommand\footnoteref[1]{\protected@xdef\@thefnmark{\ref{#1}}\@footnotemark}
\makeatother

\def\BibTeX{{\rm B\kern-.05em{\sc i\kern-.025em b}\kern-.08em
    T\kern-.1667em\lower.7ex\hbox{E}\kern-.125emX}}
\begin{document}
\title{On the equivalence of direct and indirect data-driven predictive control approaches}

\author{Per Mattsson, Fabio Bonassi, Valentina Breschi and  Thomas B. Sch{\"o}n\thanks{This work is partly sponsored by \emph{Kjell och M{\"a}rta Beijer Foundation} and the \emph{Swedish Research Council (VR)} under contracts: 2016-06079, 2023-04546.}
\thanks{P. Mattsson, F. Bonassi and T. B. Sch{\"o}n are with Uppsala University, Uppsala, Sweden (e-mail: per.mattsson@it.uu.se, fabio.bonassi@it.uu.se, thomas.schon@it.uu.se)}
\thanks{V. Breschi is with Eindhoven University of Technology, Eindhoven, The Netherlands (e-mail: v.breschi@tue.nl)}
}

\maketitle
\thispagestyle{empty}
\begin{abstract}
Recently, several direct Data-Driven Predictive Control (DDPC) methods have been proposed, advocating the possibility of designing predictive controllers from historical input-output trajectories without the need to identify a model. In this work, we show their equivalence to a (relaxed) indirect approach, allowing us to reformulate direct methods in terms of estimated parameters and covariance matrices. This allows us to provide further insights into how these direct predictive control methods work, showing that, for unconstrained problems, the direct methods are equivalent to subspace predictive control with a reduced weight on the tracking cost, and analyzing the impact of the data length on tuning strategies. 
Via a numerical experiment, we also illustrate why the performance of direct DDPC methods with fixed regularization tends to degrade as the number of training samples increases.

\end{abstract}


\DraftwatermarkOptions{%
 angle=0,
 hpos=107.5mm,
 vpos=270.5mm,
 fontsize=2.8mm,
 color={[gray]{0.2}},
 text={
   \parbox{1\textwidth}{{\textcopyright} 2024 IEEE.  Personal use of this material is permitted.  Permission from IEEE must be obtained for all other uses, in any current or future media, including reprinting/republishing this material for advertising or promotional purposes, creating new collective works, for resale or redistribution to servers or lists, or reuse of any copyrighted component of this work in other works. \textbf{Please cite the published article instead of this manuscript.  DOI:} \href{https://doi.org/10.1109/LCSYS.2024.3403473}{10.1109/LCSYS.2024.3403473}
   }}}

\section{Introduction}
\label{sec:introduction}

\IEEEPARstart{I}{n} recent years there has been an increasing interest in data-driven control in general, and Data-Driven Predictive Control (DDPC) in particular  \cite{dorfler2023data1}. In the realm of predictive control, the interest in direct data-driven solutions has remained practically dormant from the advent of Subspace Predictive Control (SPC) \cite{favoreel1999spc} until the two seminal works in \cite{coulson2019data,berberich2020data}. Since the latter approaches rely on data-driven predictors 
derived in a noise-free setting, 
regularization terms are added to the standard predictive cost to handle noise in the data.  
On the one hand, this growing interest in DDPC has prompted research to modify the schemes in \cite{coulson2019data,berberich2020data} in different directions, e.g., taking nonlinearities into account directly as in \cite{berberich2020trajectory}, using second-order moments to average over noisy data~\cite{ferizbegovic2021willems}, exploiting LQ-factorization as in $\gamma$-DDPC \cite{breschi2023data} or enforcing causality as in~\cite{sader2023causality}.

Concurrently, this has stimulated a debate on the differences, benefits, and drawbacks of \emph{direct} methods versus \emph{indirect} methods. In the former, a controller is designed directly from data, whereas in the latter, a parametric model is first estimated and then used for control design (see, e.g., \cite{dorfler2023data2, krishnan2021direct, dorfler2022bridging}). It is worth pointing out that the line between direct and indirect methods is (and has always been) blurred. For example, SPC is traditionally regarded as a direct approach as it only inherits the projection steps of subspace identification methods \cite{hou2013classification}. At the same time, in the recent literature, it is often labeled as an indirect approach as it exploits a predictor that directly stems from identification arguments (see, e.g., \cite{Verheijen2023review,dorfler2022bridging}). Nonetheless, recent works have shown the tight connection between DeePC and SPC. Indeed, DeePC and SPC are proven equivalent for noise-free LTI systems in \cite{fiedler2021relationship}. When the measured outputs contain noise, \cite{mattsson2023regularization} shows that the regularization in DeePC penalizes deviations from the SPC predictor. Expressions for the implicit predictor in DeePC are derived in \cite{kladtke2023implicit}, again showing a close relation to the SPC predictor. Instrumental Variables (IV) are introduced into the DeePC method in \cite{van2022data, wang2023data}. These IVs remove the need for extra regularization, and it is seen that the resulting direct DDPC is equivalent to an indirect IV-SPC method. Hence, when there is no noise, or instrumental variables are used instead of regularization to handle noise, direct DDPC methods are equivalent to SPC methods.   

In this contribution, we show that several direct DDPC methods are equivalent to an indirect method with a slack variable in the predictor. {This equivalence holds for finite data regimes \emph{without any assumptions} on the available training data (e.g. persistent exciting inputs or noise distributions).
The derived indirect reformulation can be used to explain, in a structured way, the role played by the length of the training dataset on the performance of direct DDPC methods with fixed regularization weights.} At the same time, such a reformulation sheds further light on how and why direct DDPC approaches work, allowing us to analyze potential drawbacks
due to the implicit presence of the predictor slack. This is exemplified by noting that, for unconstrained problems with 
quadratic cost, most existing direct DDPC approaches are equivalent to SPC with a reduced weight on the tracking cost.

The remainder of the paper is organized as follows. Section~\ref{sec:problem} introduces the predictive control problem, while Section~\ref{sec:background} gives a background on the considered data-driven methods. The equivalence between the direct and indirect methods is presented in Section~\ref{sec:result}, with implications discussed in Section~\ref{sec:implications} and numerically illustrated in Section~\ref{sec:example}. We conclude the paper in Section~\ref{sec:conclusions}.

\subsection{Notation}
For a vector $x$ we define the $\ell_2$-norm as $\| x\|_2^2 = x^\top x$, and the weighted $\ell_2$-norm as $\| x \|_W^2 = x^\top W x$. For a matrix $M$ we let $\range(M)$ denote the range space of the matrix, $M^\dagger$ its Moore-Penrose pseudo-inverse, $\Pi_{M} = M^\dagger M$ the orthogonal projector onto $\range(M^\top)$, so $I - \Pi_M$ is the orthogonal projector onto the null space of $M$. 
We use bold-faced letters like $\yf$ for vectors containing the value of a signal at several time steps.  

\clearpage

\section{Problem setting}\label{sec:problem}
We consider the predictive control problem for a system with inputs $u_t \in \mathbb{R}^{n_u}$ and outputs $y_t \in \mathbb{R}^{n_y}$. Since we are studying data-driven control schemes like \cite{coulson2019data}, we assume that at time $t$ the $\rho$ past inputs and outputs 
\begin{equation} \label{eq:zvec}
   \zinit_t = \begin{bmatrix} u_{t-\rho}^\top & \cdots & u_{t-1}^\top & y_{t-\rho}^\top & \cdots & y_{t-1}^\top \end{bmatrix}^\top,
\end{equation}
have been measured. Denote $T$ future inputs and outputs as
\begin{equation}
    \uf_t = \begin{bmatrix} u_t^\top \, \cdots \, u_{t+T-1}^\top \end{bmatrix}^\top \!\!\!\!, \quad   \yf_t = \begin{bmatrix} y_t^\top \, \cdots \,  y_{t+T-1}^\top \end{bmatrix}^\top\!\!\!\!.\label{eq:inoutvectors}
\end{equation}
The aim is to minimize some finite-horizon criterion $J(\uf_t, \yf_t)$, for example the quadratic criterion
\begin{equation}\label{eq:Jex}
J(\uf_t, \yf_t) = \| \yf_t - \bar{\yf} \|_{Q}^2 + \| \uf_t - \bar{\uf} \|_{R}^2,
\end{equation}
where $(\bar{\yf}, \bar{\uf})$ denote a reference trajectory and $Q$ and $R$ are user-defined positive definite weight matrices.

The measured outputs in $\zinit_t$ are in practice corrupted by noise, and we do not assume access to a perfect model of the system. Therefore, we  rely on the predicted outputs
\[
    \yh_t = \begin{bmatrix} \hat{y}_t^\top & \cdots & \hat{y}_{t+T-1}^\top \end{bmatrix}^\top,
\]
and formalize the finite-horizon predictive control problem as\footnote{When there is no risk for confusion, we drop the subscript $t$ on $\zinit, \uf, \yf, \yh$.}
\begin{equation}\label{eq:opt}
\begin{aligned}
    \min_{\uf}&\quad J(\uf, \yh) \\
    \text{s.t.}&\quad \yh = \text{Prediction}(\zinit, \uf), \\
              &\quad \uf \in \mathcal{U}, \quad \yh \in \mathcal{Y},
\end{aligned}
\end{equation}
where $\mathcal{U}$ and $\mathcal{Y}$ are constraints on the future inputs and outputs. 

\section{Background}\label{sec:background}
Before stating the main result of the paper we give a brief background on the direct DDPC methods considered, as well as a related multi-step predictor model.

\subsection{Multi-step prediction model}\label{sec:indirect}
A straightforward way to perform the multi-step predictions $\yh$ is to consider a linear in the parameters model 
\begin{equation}\label{eq:mpred}
    \yh = \That \varphi(\zinit, \uf),
\end{equation}
where $\varphi(\zinit, \uf) \in \real^{n_{\varphi}}$ is a regression vector and $\That \in \real^{n_y T \times n_{\varphi}}$ are estimated parameters. For notational convenience, we will often drop the arguments of $\varphi$, and focus the discussion on 
\begin{equation}\label{eq:philinear}
    \varphi = \begin{bmatrix} \zinit^\top & \uf^\top \end{bmatrix}^\top, 
\end{equation}
which is reasonable at least when the system to be controlled is (approximately) linear and $\rho$ is large enough. However, $\varphi$ can be a known transformation of $\zinit$ and $\uf$, e.g., the 
basis expansion used in \cite{berberich2020trajectory} for direct data-driven analysis of nonlinear systems.

To estimate $\That$, the data $\{ (\yf_t^d, \varphi_t^d)\}_{t=1}^N$ are used, where the superscript $d$ denotes training data. Let 
\begin{align}
    Y  &= \begin{bmatrix} \yf_1^d & \cdots & \yf^d_{{N}} \end{bmatrix} \in \real^{n_y T \times N}, \label{eq:Y} \\
    \Phi &= \begin{bmatrix} \varphi_1^d & \cdots & \varphi^d_{{N}} \end{bmatrix} \in \real^{n_\varphi \times N}, \label{eq:Phi}
\end{align}
and define
\begin{align}
    \covpred &= \frac{1}{N} (Y-\That\Phi)(Y-\That\Phi)^\top, \label{eq:covpred} \\
    \covphi &= \frac{1}{N} \Phi\Phi^\top, \label{eq:covphi}
\end{align}
where $\covpred$ can be seen as an estimated prediction error covariance, and $\covphi$ as an estimate of the second order moments of~$\varphi_t^d$. 
Given this, a standard method for parameter estimation is to {minimize the squared prediction errors}
\begin{equation}\label{eq:estcrit}
  \trace(\covpred) = \frac{1}{N} \sum_{t=1}^N (\yf_t^d - \That \varphi_t^d)^\top (\yf_t^d - \That \varphi_t^d).
\end{equation}
If $\covphi$ is positive definite, the unique minimizer is given by 
\begin{equation}\label{eq:That}
    \That = Y \Phi^\dagger.
\end{equation}
Otherwise, there are infinitely many minimizers and, among them, \eqref{eq:That} is the one having the minimum Forbenious norm.

\begin{remark}\label{rem:N} When using for example \eqref{eq:philinear}, we need $\rho$ data-points before $t=1$ and $T-1$ data points after $t=N$ to construct all $\varphi_i^d$. That is, a total of $\wbar{N}=\rho+T+N-1$ training data-points are needed to construct $Y$ and $\Phi$.
\end{remark}
\begin{remark}[SPC]\label{rem:spc}
When the estimate \eqref{eq:That} is used for multi-step predictions in \eqref{eq:opt} we get the SPC method \cite{favoreel1999spc}.
\end{remark}

\subsubsection{{Causal multi-step prediction model}}\label{sec:indirect:causal}
As will be shown in Section~\ref{sec:result}, many direct DDPC strategies implicitly use
the estimate in \eqref{eq:That}.
However, this estimate does not enforce any structure on $\That$, so $\hat{y}_{t+k}$ can be affected by $u_{t+\ell}$ for $\ell > k$ unless the corresponding elements in $\That$ are estimated to zero. Hence, the predictions made by \eqref{eq:mpred} are in general not causal. It is possible to get causal predictions by enforcing a structure on $\That$, see e.g., \cite{qin2005novel, mattsson2023regularization, sader2023causality}. This will reduce the number of parameters to estimate, 
and often leads to improved predictions, see e.g.~\cite{mattsson2023regularization} 
{and the example in Section~\ref{sec:example}}.

\subsection{Direct methods}\label{sec:direct}
In direct methods, no explicit prediction model is estimated. 
Instead, predictions are made by finding a linear combination of $\varphi_t^d$ in the training data (i.e., the columns of $\Phi$) that equals the new $\varphi$, and then let the prediction $\yh$ be the same linear combination of training outputs $\yf_t^d$ (i.e., the columns of $Y$). That is, we find a {$g \in \mathbb{R}^{N}$} that solves
\begin{subequations} \label{eq:PhiYg} 
\begin{align}
    \varphi(\zinit, \uf) &= \Phi g, \label{eq:Phig}\\
    \yh &= Y g. \label{eq:Yg} 
\end{align}
\end{subequations}
Note that \eqref{eq:Phig} can have infinitely many solutions $g$ for a fixed $\varphi$ if $N>n_{\varphi}$, and potentially no solutions if $N<n_{\varphi}$ or $\Phi$ is rank deficient. 
As long as \eqref{eq:Phig} admits a solution,  $g^*$ is a solution if and only if there exists {$ w \in \mathbb{R}^{N}$} such that
\begin{equation}\label{eq:gsol}
g^* = \Phi^\dagger \varphi + (I-\Pi_\Phi)w.
\end{equation}
{The predicted outputs are thus given by
\begin{equation}\label{eq:mfpred}
\yh = Y\Phi^\dagger \varphi + Y(I-\Pi_\Phi)w = \That \varphi + \Delta \yh,
\end{equation}
where $\Delta \yh = (Y-\That\Phi)w \in \range(Y-\That\Phi) = \range(\covpred)$. That is, the predictions made by \eqref{eq:PhiYg} can be viewed as a modified version of the multi-step predictor \eqref{eq:mpred} which adds a slack variable that is constrained to lie in the subspace spanned by the prediction errors on the training data.} 

{When the training data contains noise and $N>n_\varphi$, it is typically not possible to make $\covpred =0$, so \eqref{eq:PhiYg} admits infinitely many solutions.} In direct DDPC methods this is often addressed by introducing a regularization on $g$, cf. \cite{dorfler2022bridging, berberich2020data}. The resulting direct DDPC method can thus be formulated as
\begin{equation}\label{eq:deepc}
\begin{aligned}
    \min_{g}&\quad J(\uf, \yh) + \beta h(g), \\
    \text{s.t.}&\quad \Phi g = \varphi(\zinit, \uf), \\
              &\quad Y g = \yh, \\
              &\quad \uf \in \mathcal{U}, \quad \yh \in \mathcal{Y},
\end{aligned}
\end{equation}
where $\beta$ is a hyper-parameter and $h(g)$ is some regularizer. {In this contribution, we restrict our analysis to}
\begin{align}
    h(g) &= \| g \|_2^2, \label{eq:l2reg} \\
    h(g) &= \| (I-\Pi_\Phi) g \|_2^2, \label{eq:pl2reg}
\end{align}
where the first is simply the $\ell_2$-norm, and the second is called an ``identification induced regularizer'' in \cite{dorfler2022bridging}, {leaving the case of e.g. the $\ell_1$-norm regularizers {\cite{coulson2019data}} to future works.}
\begin{remark}\label{rem:rank}
    If $\Phi$ does not have full row rank,  then there exists $\zinit$ for which \eqref{eq:deepc} has no feasible solution. A possible modification is to add slack variables to $\varphi$, akin to
    what was done in 
    \cite{coulson2019data, berberich2020data, fiedler2021relationship}. To avoid complicating the notation, we do not include such slack variables in this contribution.
\end{remark}
\begin{remark}
For simplicity, $Y$ and $\Phi$ have here been constructed as Hankel matrices, c.f.  \eqref{eq:philinear}-\eqref{eq:Phi}.
The result of Section~\ref{sec:result}, however, is also valid for other structures of $Y$ and $\Phi$, like Page  \cite{Verheijen2023review}, second order moments-based  \cite{ferizbegovic2021willems}, and Koopman-based \cite{lian2021koopman, lazar2023basis} matrices, as long as predictions are made according to \eqref{eq:PhiYg}. 
\end{remark}

\subsubsection{{$\gamma$-DDPC formulation}}
In \cite{breschi2023data}, an alternative formulation of~\eqref{eq:deepc}, called $\gamma$-DDPC, is derived under the assumption:
\begin{assumption}\label{assumption}
    $\covpred$ and $\covphi$ in \eqref{eq:covpred}-\eqref{eq:covphi} are positive definite when $\That = Y\Phi^\dagger$.
\end{assumption}
Note that this assumption is equivalent to assuming that $\begin{bmatrix} \Phi^\top \!\!&\!\! Y^\top \end{bmatrix}^\top$ has full row rank, and holds almost surely if $N>n_{\varphi}+n_yT$, the training data has persistently exciting inputs and white measurement noise in the outputs, cf. \cite{breschi2023data}.

To derive the method, we introduce the Hankel matrices $Z$ and $U$ which are the upper an lower part of $\Phi$ respectively, and the LQ-decomposition
\begin{equation}\label{eq:LQ}
    \begin{bmatrix} Z \\ U \\ Y \end{bmatrix} = \underbrace{\begin{bmatrix} L_{11} & 0 & 0 \\ L_{21} & L_{22} & 0 \\ L_{31} & L_{32} & L_{33} \end{bmatrix}}_{L} \underbrace{\begin{bmatrix} Q_1 \\ Q_2 \\ {Q_3} \end{bmatrix}}_{Q},
\end{equation}
where $L_{ii}$ are invertible, and $Q_i Q_i^\top = I$ while $Q_i Q_j^\top = 0$ if $i\neq j$. Let $\gamma = Q g$ so that we can rewrite \eqref{eq:PhiYg} as 
\[
    \begin{bmatrix} \varphi \\ \yh \end{bmatrix} = \begin{bmatrix} \zinit \\ \uf \\ \yh \end{bmatrix} =  L \gamma. 
\]
Compared to \eqref{eq:deepc}, $\gamma$ has a lower dimension than $g$, and due to the structure of $L$ we can interpret what different parts of $\gamma$ correspond to. Indeed, we can divide $\gamma$ into $\gamma_1 = L_{11}^{-1} \zinit$ that relates to initial values, $\gamma_2$ that relates to future inputs, and $\gamma_3$ that is linked to the quality of the predictions. The formulation of $\gamma$-DDPC in \cite{breschi2023impact} is:
\begin{equation}\label{eq:gddpc}
    \begin{aligned}
    \min_{\gamma}&\quad J(\uf, \yh) + \beta_2 \| \gamma_2\|_2^2 + \beta_3 \| \gamma_3\|_2^2,\\
                 &\quad \begin{bmatrix} \uf \\ \yh \end{bmatrix} = \begin{bmatrix} L_{21} & L_{22} & 0 \\ L_{31} & L_{32} & L_{33} \end{bmatrix} \begin{bmatrix} L_{11}^{-1} \zinit \\ \gamma_2 \\ \gamma_3 \end{bmatrix}, \\
              &\quad \uf \in \mathcal{U}, \quad \yh \in \mathcal{Y}.
    \end{aligned}
\end{equation}

\section{The main result}\label{sec:result}
Our main result is that all the direct methods discussed in Section~\ref{sec:direct} are equivalent to the indirect formulation
\begin{empheq}[box=\fbox]{equation}\label{eq:result}
    \begin{aligned}
        \min_{\uf, \Delta \yh} &\quad J(\uf, \yh) + \frac{\lambda_1}{N} \| \varphi(\zinit, \uf) \|_{\covphi^\dagger}^2 + \frac{\lambda_2}{N} \| \Delta \yh \|_{\covpred^{\dagger}}^2, \\
                \text{s.t.} &\quad \yh = \That \varphi(\zinit, \uf) + \Delta \yh, \\
                            &\quad \Delta \yh \in \range(\covpred), \quad \varphi(\zinit, \uf) \in \range(\covphi), \\
                            &\quad \uf \in \mathcal{U}, \quad \yh \in \mathcal{Y},
    \end{aligned}
\end{empheq}
{where $\Delta \yh$ is introduced in \eqref{eq:mfpred}}, $\lambda_1$ and $\lambda_2$ are hyper-parameters, and $\covpred$ and $\covphi$ are defined in \eqref{eq:covpred}-\eqref{eq:covphi}.
{For the equivalence to hold no assumptions on the training data is required, but if Assumption~\ref{assumption} holds the range constraints can be omitted.}

\begin{theorem}\label{thm:result}
    If \eqref{eq:deepc} is feasible, it has the same optimal $\uf$ and $\yh$ as \eqref{eq:result} with $\That = Y\Phi^\dagger$ and
    \begin{itemize}
        \item $\lambda_1 = \lambda_2 = \beta$ if $h(g) = \|g\|_2^2$.
        \item $\lambda_1 =0$ and $\lambda_2 = \beta$ if $h(g) = \|(I-\Pi_\Phi) g\|_2^2$.
    \end{itemize}
    Furthermore, if Assumption~\ref{assumption} holds then \eqref{eq:result} and \eqref{eq:gddpc} have the same optimal $\uf$ and $\yh$ if $\lambda_1 = \beta_2$ and $\lambda_2 = \beta_3$.
\begin{proof} 
See Appendix~\ref{app:thm1}.
\end{proof}\end{theorem}
\begin{remark}
{Compared to the standard 
    formulation \eqref{eq:deepc} 
    having $N$ optimization variables, 
    $\gamma$-DDPC and the indirect formulation \eqref{eq:result} only have $(n_u\!+\!n_y)T$ optimization variables, making them appealing for large datasets where $N \gg T$. }
\end{remark}
\begin{remark}
In \cite{sader2023causality} a causal variant of $\gamma$-DDPC, named RC-DeePC, is proposed.
In Appendix~\ref{appendix:gamma} we show that it is equivalent to~\eqref{eq:result} with $\That$ replaced by the casual estimate given in \cite[Equation 17]{sader2023causality}.
cf. Section~\ref{sec:indirect:causal}.
\end{remark}
\begin{remark}
If $\covphi$ is not positive definite~\eqref{eq:result} can be infeasible, cf. Remark~\ref{rem:rank}. As in \cite{coulson2019data, fiedler2021relationship}, slack variables can be added in $\varphi$ to handle this. With~\eqref{eq:result} this can be avoided by simply removing the range constraint on $\varphi$, since the predictions $\That \varphi = Y\Phi^\dagger \varphi = Y\Phi^\dagger (\Phi\Phi^\dagger \varphi)$ are only affected by $\Phi \Phi^\dagger\varphi$, i.e., the orthogonal projection of $\varphi$ onto $\range(\covphi)$. 
\end{remark}

\section{Implications}\label{sec:implications}
The main result outlined in Section~\ref{sec:result} allows us to analyze direct DDPC approaches in terms of estimated parameters, covariance matrices, and slack variables. 

\subsection{The role of the slack variable} \label{sec:implications:slack}

As the slack variable in \eqref{eq:result}
can adjust the predicted outputs in $J$ without changing the inputs $\uf$, the effect is similar to reducing the weight on the tracking cost in the criterion. 
This can be seen explicitly for unconstrained problems.
\smallskip
\begin{corollary}\label{corr:unconstrained}
Consider the \emph{unconstrained problem}, where the constraint sets $\mathcal{U}$ and $\mathcal{Y}$ are removed from \eqref{eq:result}, and let $J$ be the quadratic cost function in \eqref{eq:Jex}. If  $\covpred$ and $\covphi$ are positive definite, then \eqref{eq:result} give the same optimal $\uf$ as
\begin{equation}\label{eq:unconstrained}
\min_{\uf} \quad \| \bar{\yf} - \That\varphi(\zinit, \uf)\|_{\wtilde{Q}}^2 + \| \bar{\uf} - \uf\|_R^2 + \frac{\lambda_1}{N}\| \varphi(\zinit, \uf)\|_{\covphi^{-1}}^2,
\end{equation}
where $\wtilde{Q} = Q - Q(Q+\frac{\lambda_2}{N}\covpred^{-1})^{-1}Q \preceq Q$.
\end{corollary}
\begin{proof}
In the unconstrained problem we can solve for $\Delta\yh$ to get $\Delta\yh = (Q+\frac{\lambda_2}{N} \covpred^{-1})^{-1}Q(\bar{\yf} - \That\varphi)$. Inserting this back into the criterion of \eqref{eq:result} yields \eqref{eq:unconstrained}.
\end{proof}
\smallskip
That is, unconstrained direct DDPC methods with a quadratic cost function are equivalent to SPC with a reduced weight on the tracking cost. Hence, in the extreme case when $\lambda_2=0$ the tracking cost is disregarded by direct DDPC methods. It can be noted that the addition of $\| \varphi \|_{\covphi^\dagger}^2$ increases the control cost, which has a similar effect. 
In the presence of constraints on $\yh$ the slack variable can also adjust the predictions to satisfy the constraints without adjusting the input. Hence, the direct methods re-weights SPC to focus more on the control cost.

\subsection{The role of the past horizon $\rho$}
\label{sec:implications:past}
When $\varphi$ is chosen as \eqref{eq:philinear} we can see $\That \varphi$ as a multi-step ARX predictor, where e.g., $\hat{y}_{t}$ is given by a 1-step ARX predictor. From the system identification literature, it is well known that for LTI systems with measurement noise---which is a case often studied in the direct DDPC literature---the optimal predictor has an output error structure and can not be expressed as an ARX predictor with finite $\rho$, see e.g., \cite{ljung1999system}. This implies that it can be beneficial to increase $\rho$ beyond the order of the underlying system. However, increasing $\rho$ can also lead to more variance in the estimated parameters, as empirically analyzed in e.g., \cite{breschi2023data, breschi2023subspace}.
In light of the reformulation \eqref{eq:result}, we can 
further note that increasing $\rho$ on a fixed training data set leads to a reduction of $\trace(\covpred)$, and thus the penalty on the slack variable tend to increase.
For $\rho$ large enough, $\covpred=0$, {which means that the slack $\Delta \yh$ is constrained to zero}, and the predictors used by direct DDPC methods coincide with the multi-step predictor in \eqref{eq:mpred}.

\subsection{The role of the training data} \label{sec:implications:training}
For this section, let us assume that the training data has persistently exciting inputs and the measured outputs are corrupted by white noise, so $\Phi$, $Y$ and $\begin{bmatrix} \Phi^\top \!\!&\!\! Y^\top\end{bmatrix}^\top$ have full rank almost surely, cf. \cite{breschi2023data}, even though the number of samples must be large enough to get full \emph{row} rank. We can then see that the behavior of the direct DDPC methods go through some transition as $N$ increases:
\begin{itemize}
    \item $N \leq n_\varphi$: there exists $\That$ that fits the training data perfectly, so $\covpred =0$ and the direct methods use the multi-step predictor in \eqref{eq:mpred} without slack. 
    \item $n_\varphi < N < n_\varphi+n_yT$: the range space of $\covpred$ increases until $\covpred$ becomes positive definite at $N = n_\varphi + n_yT$.
    \item $N > n_\varphi + n_yT$: as $N$ increases, $\covpred$ converges to the prediction error covariance using the asymptotic~$\That$. 
\end{itemize}
    Note that, the estimate $\That$ is typically poor for $N=n_\varphi$ due to its high variance. 
    {As $N$ increases, it will, under relatively weak assumptions \cite{ljung1999system}, converge to the predictor within the model class \eqref{eq:mpred} that has the smallest prediction error covariance for the system, and experimental conditions, used to generate the training data. While there may exist better predictors outside the model class \eqref{eq:mpred}, it can be expected that the predictions of \eqref{eq:mpred}, as well as the performance of SPC, will improve with increasing $N$.}
    

At the same time, while $\covpred$ typically converges to a finite positive definite matrix as $N\rightarrow \infty$, the penalty on the slack in \eqref{eq:result} is also weighted by $\lambda_2/N$. So, for fixed $\lambda_2$, \eqref{eq:result}  will allow for more slack when $N$ is large. This can also be seen in Corollary~\ref{corr:unconstrained}, where $\wtilde{Q}\rightarrow 0$ as $N\rightarrow \infty$, thus leading the direct methods to disregard the tracking cost for large $N$ unless $\lambda_2$ is also increased. This shows an important drawback of using the $h(g) = \|g\|_2^2$ in \eqref{eq:deepc}, since it sets $\lambda_1 = \lambda_2$. 
To achieve good predictions for large $N$ we must increase $\lambda_2$, but this will also mean that the added penalty on $\varphi$ is kept large even for large~$N$. This problem is avoided by the $\gamma$-DDPC method \cite{breschi2023data} which penalizes $\varphi$ and the slack variable separately, highlighting from a different perspective the tuning challenges inherent in the approaches proposed in \cite{coulson2019data,berberich2020data}.

\section{Illustrative example} \label{sec:example}
The implications discussed in Section~\ref{sec:implications} are now illustrated on the system considered in 
 \cite{mattsson2023regularization}, i.e., 
\begin{equation*}
    y_t = 1.2 y_{t-1} - 0.3 y_{t-2} - 0.1 y_{t-3} + 0.5 u_{t-1} - 0.4 u_{t-2} + 0.1 u_{t-3},
\end{equation*}
with an additional input constraint $u \in \mathcal{U} = [-1, 1]$. In our evaluation, we consider datasets with increasing $\wbar{N}=\rho+T+N-1$ (see Remark~\ref{rem:N}), to empirically illustrate the discussion in Section~\ref{sec:implications:slack} and \ref{sec:implications:training}. In particular, for the training data collection, the system is excited with a (saturated) Gaussian input $u \sim \mathcal{N}(0, \sigma_u^2)$, with $\sigma_u = 0.6$.
The output measurements are corrupted by white noise with standard deviation $\sigma_y = 0.1$, which yields a signal-to-noise ratio of $10$, approximately.

\let\thefootnote\relax\footnotetext{Source code: \textcolor{blue}{\url{https://github.com/bonassifabio/ddpc}}}

Potential benefits and drawbacks of the additional penalty on $\varphi$ in \eqref{eq:result} have been studied in the $\gamma$-DDPC framework~\cite{breschi2023impact}. 
Therefore, we here let $\lambda_1 =0$, and focus on the impact of the slack variable implicitly introduced by direct methods.
The following controllers have been considered:
\begin{enumerate}[a)]
    \item \textbf{DeePC}$_{\lambda_2}$ -- DeePC \eqref{eq:deepc} with regularization \eqref{eq:pl2reg}, equivalent to \eqref{eq:result} with   slack regularization weight $\lambda_2 = \beta \in \{1, 10, 100, 1000\}$;
    
    \item \textbf{SPC} -- See Remark~\ref{rem:spc}, equivalent to \eqref{eq:result} with $\That$ given by~\eqref{eq:That}, and  slack constrained to zero.
    \item \textbf{C-SPC} -- 
    SPC with $\That$ estimated with a structure that enforces causal predictions, cf. Section~\ref{sec:indirect:causal}. 
    
\end{enumerate}
For all these methods, a prediction horizon $T = 30$, a model-order $\rho =20$, and the quadratic cost function  \eqref{eq:Jex} with weights $Q=I$, $R = 0.1I$, and $\bar{u}_{t+k} =0$ have been used. 

\smallskip
\noindent{\bf Setup} -- In the tests conducted, we assume that the system has null initial conditions and the setpoint is $\bar{y}_{t+k} = 0.75$. The performance of the controllers described above is evaluated, for a given training set, by simulating the open-loop response $\tilde{\yf}(\uf^\star)$ of the system under the optimal control sequence $\uf^\star$, i.e., the optimal solution of \eqref{eq:result}.
The first performance index used to evaluate the controllers is the achieved cost criterion 
\begin{equation}\label{eq:perfcriterion}
    J^\star = J \big(\uf^\star, \tilde{\yf}(\uf^\star) \big) = \| \tilde{\yf}(\uf^\star) - \bar{\yf} \|_Q^2 +  \| \uf^\star \|_R^2.
\end{equation}
We furthermore evaluate the distance from the \emph{oracle}, i.e., a state-feedback MPC regulator based on the exact system model, having the same cost function and prediction horizon as the DDPC schemes.
This comparison evaluates how close to the ideal controller the different methods are.
Letting $\uf^o$ be the oracle's optimal control sequence and  $\yf^o$  the corresponding output,  we can evaluate the dissimilarity from the oracle as
\begin{equation} \label{eq:oracledist}
    J^o(\uf^\star, \tilde{\yf}(\uf^\star) \big) = \| \tilde{\yf}(\uf^\star) - \yf^o \|_Q^2 +  \| \uf^\star - \uf^o \|_R^2.
\end{equation}
These performance indices have been averaged over $30$ uncorrelated realizations of the measurement noise affecting $y$. 

In the evaluation, the length of the training data set has been varied from $\wbar{N} = 119$ to $\wbar{N} = 10^4$. Note that here $\wbar{N}=119$ corresponds to $N=n_\varphi$. For each $\wbar{N}$, the performance is assessed over 200 random realizations of the training data, ensuring that the results are not biased by a specific realization. 
\begin{figure}
    \centering
    \includegraphics[width=0.95 \linewidth]{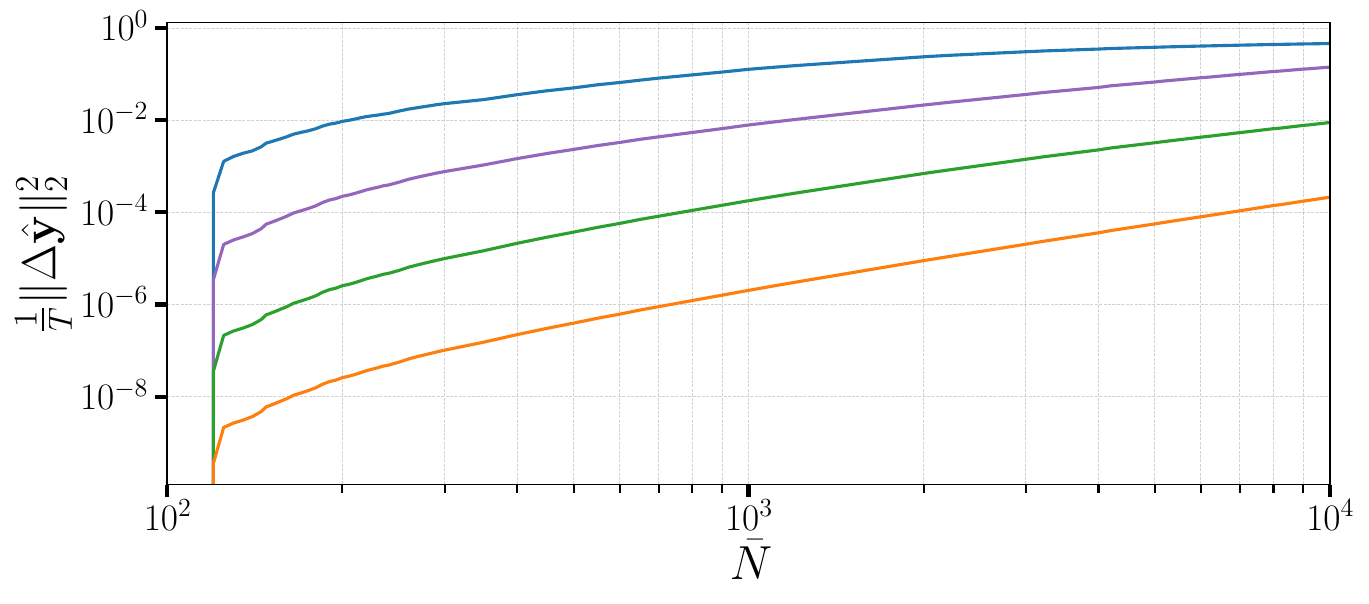}
    \vspace{-3mm}
    \caption{Mean squared value of the slack 
    $\Delta \yh$ used by DeePC$_{\lambda_2}$ for different values of 
    $\lambda_2$ (blue: $1$, purple: $10$, green: $100$,  orange: $1000$).}\vspace{-.3cm}
    \label{fig:slack}
\end{figure}
\begin{figure}
    \centering
    \hspace{1.7mm}\includegraphics[width=0.925 \linewidth]{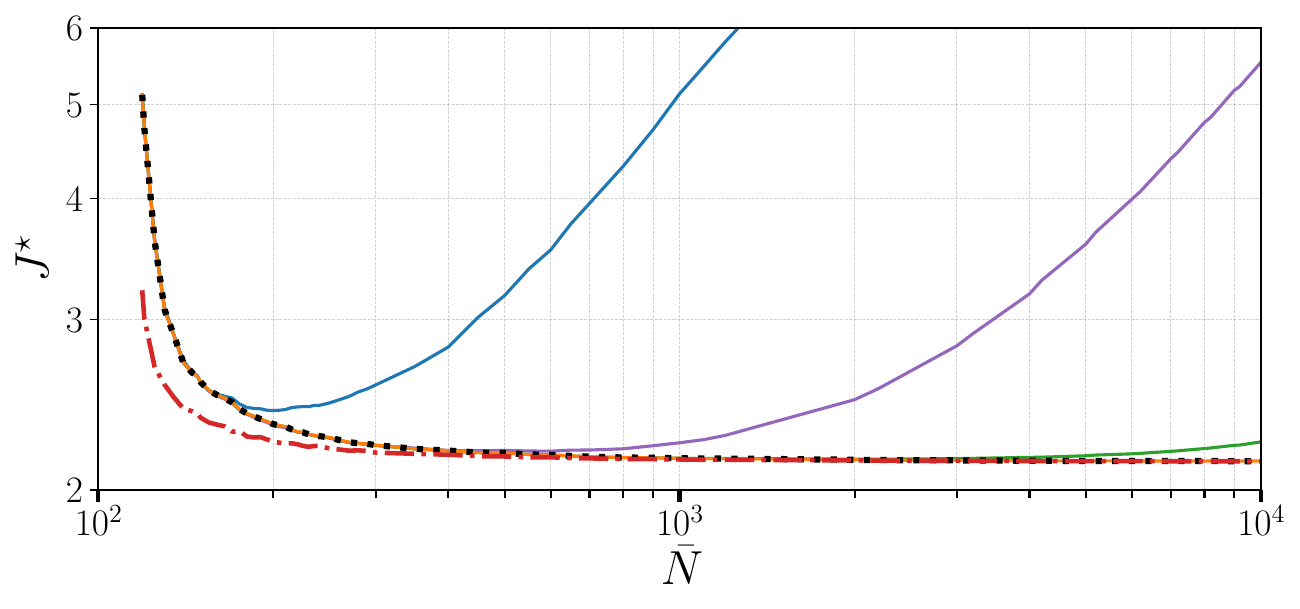} \\
    \includegraphics[width=0.95 \linewidth]{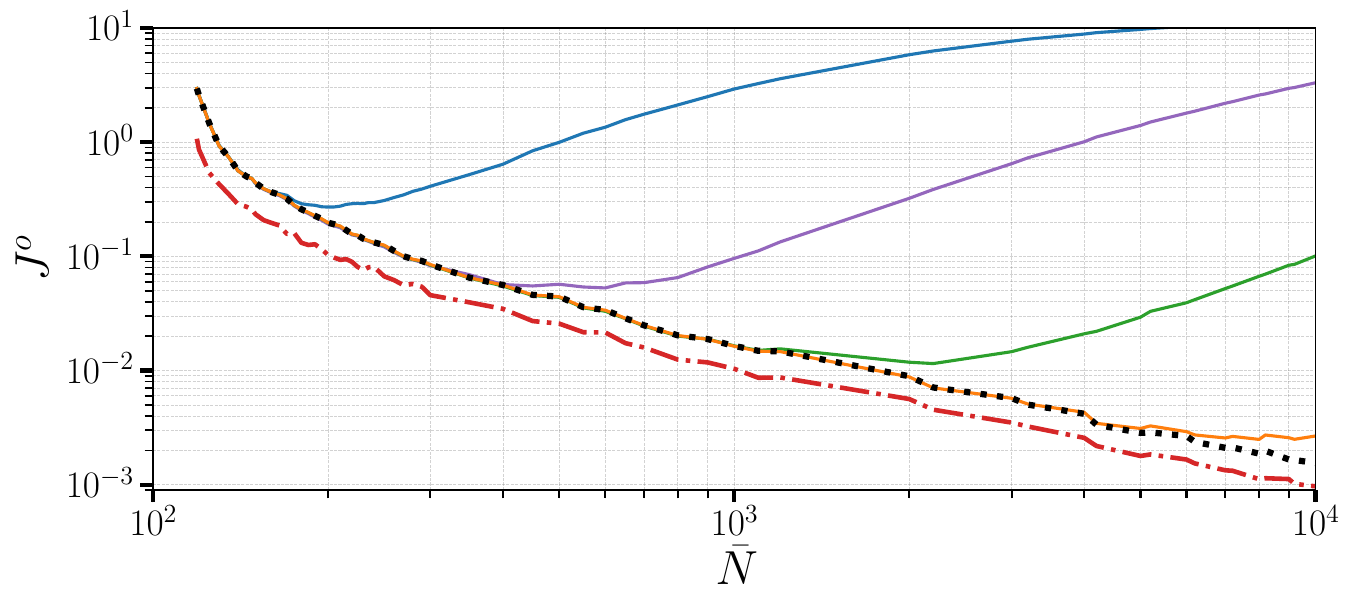}
    \vspace{-3mm}
    \caption{Cost criterion \eqref{eq:perfcriterion} (top) and dissimilarity from the oracle \eqref{eq:oracledist} (bottom) achieved by the DeePC$_{\lambda_2}$ (solid lines, $\lambda_2$ color-coded as in Figure \ref{fig:slack}), SPC (black-dotted line), and C-SPC (red dashed-dotted line).}
    \label{fig:criterion}\vspace{-.3cm}
\end{figure}

\smallskip
\noindent{\bf Results and discussion} -- In Fig. \ref{fig:slack}, we show the mean squared value of the slack variable $\Delta \yh$ used by DeePC$_{\lambda_2}$.
As discussed in Section \ref{sec:implications:training}, $\covpred = 0$ for $\wbar{N}=119$ in this setting, so the slack is {forced to zero by the range constraint.
As $\wbar{N}$ increases, the rank of $\covpred$ increases until it has full row-rank while the regularization weight $\frac{\lambda_2}{N}$ decreases. This implies that DeePC$_{\lambda_2}$ will make more use of the slack variable for a fixed $\lambda_2$, which is evident from Fig \ref{fig:slack}.
This confirms the insight, discussed in Section \ref{sec:implications:training}}, that
the user parameter $\lambda_2$ should be tuned relatively to $\wbar{N}$ to ensure that the slack variable is properly penalized.

Figure \ref{fig:criterion} shows the performance indices scored by DeePC$_{\lambda_2}$ in comparison to those achieved by SPC and by C-SPC. 
As discussed, in low-data regimes the slack is limited {by the range constraint and is penalized more since $\frac{\lambda_2}{N}$ is big. This means that the slack is used less, see again Fig.~\ref{fig:slack}, so that DeePC$_{\lambda_2}$ closely resemble SPC, aligning with the results shown in~\cite{fiedler2021relationship}.
On the other hand, as discussed in Section~\ref{sec:implications:slack} the performance of DeePC$_{\lambda_2}$ deteriorates for large $\wbar{N}$ if $\lambda_2$ is fixed, due to the increasing use of the slack variable which reduces the tracking cost. This effect is more pronounced for small $\lambda_2$.}

In the considered case, the slack variable implicitly introduced by the direct DeePC$_{\lambda_2}$ consistently led to worse performance than the slack-free SPC method. 
It can also be noted that C-SPC, which uses a reduced set of estimated parameters $\That$, slightly improves the performance compared to SPC, especially in low data regimes, which is in line with results in~\cite{sader2023causality}.

\section{Conclusions}\label{sec:conclusions}
By proving the equivalence of several data-driven predictive control approaches and relaxed indirect methods, this work sheds new light on the origin of the flexibility characterizing DDPC approaches while unveiling the origin of their possible fragility. Reformulating the direct methods in terms of identified parameters and covariance matrices also 
opens up the possibility of using tools from system identification for analysis in future work.
\bibliographystyle{IEEEtran}  
\bibliography{refs.bib}

\appendices

\section{Proof of Theorem~1}\label{app:thm1}
\subsection{Proof of equivalence with \eqref{eq:deepc}}\label{app:deepc}
\begin{lemma}\label{lem:range}
If $x \in \range(M)$ then $\|M^\dagger x\|_2^2 = \| x\|_{(MM^\top)^\dagger}^2$.
\end{lemma}
\begin{proof}
Since $x \in \range(M)$ we get $x = MM^\dagger x$. Hence, by expanding the norms and using 
$M^\dagger = M^\top(MM^\top)^\dagger$ and $M^\dagger = M^\dagger M M^\dagger$ the proposition follows. 
\end{proof}
We assume that \eqref{eq:deepc} is feasible, and use \eqref{eq:gsol} to make a change of variable from $g$ to $w$. Then
\[
    h(g) = \| (I-\Pi_{\Phi}) w \|_2^2 + \begin{cases} \frac{1}{N} \| \varphi \|_{\covphi^{\dagger}}^2 & \text{for \eqref{eq:l2reg}} \\ 0 & \text{for \eqref{eq:pl2reg}} \end{cases},
\]
where Lemma~\ref{lem:range} and the fact that $\varphi \in \range(\Phi)=\range(\covphi)$ for feasible solutions of \eqref{eq:deepc} has been used. 
Furthermore, we see from \eqref{eq:mfpred} that $w$ satisfy 
\begin{equation}\label{eq:deltayh}
\Delta\yh \triangleq \yh -\That\varphi = (Y-\That \Phi) w = Y(I-\Pi_\Phi)w.
\end{equation}
Hence $\Delta \yh \in \range(\covpred)$. 
Noting that $w$ should minimize $\|(I-\Pi_\Phi)w\|_2^2$ while satisfying \eqref{eq:deltayh}, we get
    $(I-\Pi_\Phi) w =(Y(I-\Pi_\Phi))^\dagger \Delta \yh$, so
\[
    \| (I-\Pi_{\Phi}) w\|_2^2 = \Delta \yh^\top (Y-\That\Phi)^{\dagger\top} (Y-\That\Phi)^\dagger\Delta \yh = \frac{1}{N}\| \Delta \yh\|_{\covpred^{\dagger}}^2,
\]
where the last equality follows from Lemma~\ref{lem:range}. 

So, for feasible $\yh$ and $\uf$ we get $\varphi \in \range(\covphi)$ and $\Delta \yh = \yh - \That \varphi \in \range(\covpred)$, and the optimal $g$ gives 
\[
    h(g) = \frac{1}{N} \| \Delta\yh\|_{\covpred^{\dagger}}^2 + \begin{cases} \frac{1}{N} \| \varphi\|_{\covphi^\dagger}^2 & \text{for \eqref{eq:l2reg}} \\ 0 & \text{for \eqref{eq:pl2reg}} \end{cases}.
\]
The equivalence follows by inserting this into \eqref{eq:deepc} and making a change of variables.

\subsection{Proof of equivalence with \eqref{eq:gddpc} and RC-DeePC}\label{appendix:gamma}
We note that adding $\beta_2\|\gamma_1\|_2^2$ with $\gamma_1 = L_{11}^{-1} \zinit$ to the criterion in \eqref{eq:gddpc} does not change the optimal solution. Let 
\[
    M_1 = \begin{bmatrix} L_{11} & 0 \\ L_{21} & L_{22} \end{bmatrix}, \quad P_1 = \begin{bmatrix} Q_1 \\ Q_2 \end{bmatrix}.
\]
Note that $M_1$ is invertible, $P_1P_1^\top = I$ and $M_1 P_1 = \Phi$. Given the optimal $\uf$, the optimal $\gamma_1$ and $\gamma_2$ are given by $M_1^{-1} \varphi$, and
\[
  \| \gamma_1 \|_2^2 + \| \gamma_2\|_2^2 = \varphi^\top (M_1 M_1^\top)^{-1} \varphi  = \frac{1}{N}\| \varphi \|_{\covphi^{-1}}^2,
\]
since $M_1 M_1^\top = M_1 P_1 P_1^\top M_1^\top = \Phi \Phi^\top = N\covphi$. Furthermore $\That = Y\Phi^\dagger = \begin{bmatrix} L_{31} &  L_{32}\end{bmatrix}M^{-1}$, so $\yh = \That \varphi + L_{33} \gamma_3$, and
\[
    \|\gamma_3\|_2^2 = \| \yh - \That \varphi \|_{(L_{33}L_{33}^\top)^{-1}}^2 = \frac{1}{N} \| \yh - \That \varphi\|_{\covpred^{-1}}^2
\]
since $L_{33} L_{33}^\top = Y(I-\Pi_\Phi)Y^\top = N\covpred$. Inserting these expressions back into \eqref{eq:gddpc} and making a change of variables shows the equivalence. These results are aligned with the results in \cite[Section 4]{breschi2023ifac}, and connected with the result in~\cite[Equation (15)]{kladtke2023implicit}. 

We now consider the RC-DeePC optimization problem \cite[Equation 33]{sader2023causality} with $\lambda=\mu$.  The main difference is to introduce $\gamma_2'$ and $L_{32}'$ while modifying $L_{32}$ so that (cf. \cite[Theorem 4]{sader2023causality})
\[
\yh = K^* \varphi + \begin{bmatrix} L_{32}' & L_{33} \end{bmatrix} \begin{bmatrix} \gamma_2' \\ \gamma_3 \end{bmatrix}
\]
where the causality preserving estimate $\That = K^*$ in \cite[Equation 17]{sader2023causality} is used. Following \cite[Remark 3]{sader2023causality}, we get
\[
\| \gamma_2'\|_2^2 + \|\gamma_3\|_2^2 = \frac{1}{N} \| \yh - \That \varphi\|_{\covpred^{-1}}^2.
\]
Inserting this back into \cite[Equation 33]{sader2023causality} the equivalence with \eqref{eq:result}, using $\That = K^*$, follows.

\end{document}